\documentclass[conference]{IEEEtran}
\IEEEoverridecommandlockouts
\usepackage{cite}
\usepackage{amsmath,amssymb,amsfonts,amsthm,mathrsfs,nicefrac,enumerate}
\usepackage{algorithmic}
\usepackage{graphicx}

\usepackage[dvipsnames]{xcolor}

\newtheorem{thm}{Theorem}[section]
\newtheorem{cor}[thm]{Corollary}
\newtheorem{lem}[thm]{Lemma}

\newtheorem{definition}[thm]{Definition}
\newenvironment{defn}{\begin{definition}\rm}{\end{definition}}
\newtheorem{example}[thm]{Example}
\newtheorem{remark}[thm]{Remark}

\renewcommand{\u}{\mathbf{u}}

\renewcommand{\v}{\mathbf{v}}
\newcommand{\s}{\mathbf{s}}
\newcommand{\y}{\mathbf{y}}
\newcommand{\x}{\mathbf{x}}
\newcommand{\z}{\mathbf{z}}

\newcommand{\supp}[1]{\textrm{supp}(\mathbf{#1})}

\newcommand{\ScrC}{\mathscr{C}}
\newcommand{\auth}{\text{auth}}

\def\BibTeX{{\rm B\kern-.05em{\sc i\kern-.025em b}\kern-.08em
    T\kern-.1667em\lower.7ex\hbox{E}\kern-.125emX}}
\begin{document}

\title{Authenticated partial correction over AV-MACs: \\ toward characterization and coding\\
\thanks{This material is based upon work supported by the National Science Foundation under Grant No. 2107488.}
}

\author{\IEEEauthorblockN{Duncan Koepke, Michaela Schnell, Madelyn St.Pierre, Allison Beemer}
\IEEEauthorblockA{\textit{Department of Mathematics},
\textit{University of Wisconsin-Eau Claire}\\
Eau Claire, WI, United States}
}

\maketitle

\begin{abstract}
In this paper we study $\gamma$ partial correction over a $t$-user arbitrarily varying multiple-access channel (AV-MAC). We first present necessary channel conditions for the $\gamma$ partially correcting authentication capacity region to have nonempty interior. We then give a block length extension scheme which preserves positive rate tuples from a short code with zero probability of $\gamma$ partial correction error, noting that the flexibility of $\gamma$ partial correction prevents pure codeword concatenation from being successful. Finally, we offer a case study of a particular AV-MAC satisfying the necessary conditions for partial correction.
\end{abstract}

\begin{IEEEkeywords}
arbitrarily varying multiple-access channel, capacity region, authentication, partial correction
\end{IEEEkeywords}


\section{Introduction}

An arbitrarily varying multiple-access channel (AV-MAC) combines random noise with adversarial action over a channel with multiple senders and a single receiver. 
Classical communication over AV-MACs has been studied in a variety of works, with \cite{J81,G90,AC99,PS19} focusing on the capacity region in the two-user case. The combination of these works establish the communication capacity region, notably showing that the region has nonempty interior if and only if the channel does not have a set of channel \textit{symmetrizability} properties. Symmetrizability, defined for point-to-point AVCs in \cite{CN88}, indicates that the adversary can reliably trick the receiver into decoding in error.

While symmetrizability characterizes the communication capacity of an AVC, the analogous condition of \textit{overwritability} governs the (keyless) \textit{authentication} capacity of such a  channel \cite{KK18}. Overwritability indicates that an adversary is not only able to trick the receiver into an erroneous message estimate, but that they are able to do so \textit{while remaining undetected}.
In \cite{BGKKY20}, Beemer et al. formalize an extension of {overwritability} to the AV-MAC, in a similar vein to the extension of symmetrizability for communication. They show that the capacity region for (keyless) authentication over an AV-MAC is equal to that for communication with no adversary, provided that the channel is not overwritable to any degree.

Other work related to authentication over an AV-MAC includes work on MACs with byzantine users. In \cite{SBDP19,SBDP23-2}, Sangwan et al. consider a byzantine user in a two-user MAC, proving results on the authenticated communication capacity region. Indeed, the inner bound of the authentication capacity region for two users over an AV-MAC in \cite{BGKKY20} was accomplished by fixing the byzantine user in an an extension to three users of the scheme of \cite{SBDP19}. In general, however, the question of a byzantine user in a MAC is distinct from an AV-MAC, where the adversary's identity is known a priori.

In the present work, we give the extensions of authentication results from \cite{BGKKY20} to the case of an arbitrary number of users, then quickly turn our attention the idea of \textit{$\gamma$ partial correction}. Partial correction over an AV-MAC was introduced in \cite{BGKKY20} to bridge the gap between total correction and authentication. In contrast to pure authentication, $\gamma$ partial correction requires that a $\gamma$ fraction of users' messages be decoded correctly, even if the remainder be discarded. The particular users who are decoded accurately may change with each transmission: that is, the subset of users to be decoded is not fixed ahead of time. When $\gamma=0$, this reduces to authentication, while with $\gamma = 1$ the goal becomes classical communication. 

To our knowledge, partial correction over an AV-MAC has only been studied in \cite{BGKKY20}. However, we note that there may be a connection to list decoding over AV-MACs (see e.g. \cite{N13,C16}), wherein partial correction would require that elements in the output list match on a certain number of users. In \cite{BGKKY20}, the authors focus on the two-user case, and give some initial results showing that it is possible for a channel to have  a $\gamma$ partially correcting authentication capacity region with nonempty interior. Here, we extend these results to an arbitrary number of users. We give a set of necessary symmetrizability/overwritability conditions for $\gamma$ partial correction, present a case study for a particular channel satisfying these conditions, and provide a scheme to extend block length that preserves positive rate tuples.

Necessary background and notation is introduced in Section \ref{sec:prelims}. In Section \ref{sec:cap_region}, we present a set of necessary channel conditions for partial correction capacity regions with nonempty interior. Section \ref{sec:extension_scheme} discusses a general method for extending the block length of a short block length code tuple with desirable partial correction properties, and Section \ref{sec:case_study} provides a case study of the construction of these short codes for a particular channel. Section \ref{sec:conclusion} concludes the paper.


\section{Preliminaries}
\label{sec:prelims}

Let $[n]:=\{1,2,\ldots,n\}$, and let  $\supp{x}\subseteq [n]$ denote the support of a length-$n$ vector $\x$. Capital letters (e.g. $X$) will denote random variables, script letters the alphabets they are taken from ($\mathcal{X}$), and lower case letters their realizations ($x$).

Our setting will be a $t$-user AV-MAC, where $t\geq 2$. More specifically, a $t$-user AV-MAC is defined by a distribution $W_{Y|X_1\cdots X_t S}$, where legitimate channel inputs $X_{j}$ are taken from alphabet $\mathcal{X}_{j}$ for each $j\in [t]$, the adversary's choice of channel \textit{state} is $S\in \mathcal{S}$, and the channel output is given by $Y\in \mathcal{Y}$. In our model, we assume the adversary has full knowledge of the channel statistics and all user encoding strategies, but that $S$ is independent of the particular message sequence transmitted in any given time instance.
We begin by extending the definitions of \cite{BGKKY20}.

\begin{defn}
An \textit{$(M_{1},\ldots,M_{t},n)$ authentication code} for a $t$-user AV-MAC is given by encoders $f_{1},\ldots , f_{t}$ and decoder $\phi$:
\begin{align}
    f_{i}&: [M_{j}] \to \mathcal{X}_{j}^{n}, \ 1\leq j\leq t \\
    \phi&: \mathcal{Y}^{n} \to \left([M_{1}]\cup \{0\}\right) \times \cdots \times \left([M_{t}]\cup \{0\}\right), \label{eqn:decoder}
\end{align}
where an output of ``0'' in any coordinate indicates adversarial interference.
\end{defn}

We will sometimes directly discuss the codebooks $C_j = f_j(M_j) \subseteq \mathcal{X}_{j}^{n}$ in later sections. 
In this paper, we will be concerned primarily with correcting some portion of the users' messages, even if others must be discarded due to adversarial interference.

\begin{defn}
\label{def:partial_corr}
Let $\gamma \in (0,1)$. We say that an $(M_{1},\ldots,M_{t},n)$ authentication code for a $t$-user AV-MAC is \textit{$\gamma$ partially correcting} if, with high probability in $n$, we can correct at least $\lceil \gamma t\rceil$ of the $t$ messages.
\end{defn}

We observe that the case where $\gamma=0$ reduces to the classical notion of authentication for an AV-MAC \`{a} la \cite{SBDP19,BGKKY20,SBDP23-2}, while $\gamma=1$ bridges the gap to total correction of all user messages. The use of the open interval in Definition \ref{def:partial_corr} excludes the cases where no messages are corrected, or all are; neither of these is \textit{partial} correction. It is straightforward that if an authentication code is $\gamma$ partially correcting, then it is $\lambda$ partially correcting for all $0<\lambda<\gamma$.

Let $\phi^{-1}(A)\subseteq \mathcal{Y}^{n}$ represent the set of channel outputs which decode to some element $(i_{1},\ldots,i_{t})$ in the set $A$ under the decoder $\phi$, and let $\phi^{-1}(A)^{c}$ be the complement in $\mathcal{Y}^{n}$ of this set. Let $\x_{j}(i):=f_{j}(i)$ denote the length-$n$ encoding of message $i$ by user $j$.
Correspondingly, we let $\mathbf{i}$ denote a tuple of transmitted messages from $[M_1]\times\cdots\times [M_t]$, and $\mathbf{x}(\mathbf{i})$ its encoding under $(f_1, f_2, \ldots f_t)$.
Given a tuple of transmitted messages, $\mathbf{i}$, and adversarial state $\s$, where $\s=\s_{0}$ denotes that the no-adversary state sequence, we define the \textit{probability of
$\gamma$ partial correction error}  for the authentication code $(f_{1},\ldots,f_{t},\phi)$ by:
\begin{equation}
    e_{\gamma}(\mathbf{i},\s_0)=W(\phi^{-1}(\{\mathbf{i}\})^{c}\mid \x(\mathbf{i}), \s_0),
\end{equation}
and, when $\s\neq\s_0$,
\begin{equation}
    e_{\gamma}(\mathbf{i},\s)=W(\phi^{-1}(A_\mathbf{i})^{c}\mid \x(\mathbf{i}), \s),
\end{equation}
where $A_\mathbf{i}=\{\hat{\mathbf{i}}: \hat{i}_{j}\in\{0,i_{j}\} \text{ for }  j\in[t], |\supp{\hat{\mathbf{i}} }|\geq \lceil\gamma t\rceil\}$. That is, $A_\mathbf{i}$ is the set of decoded sequences that match sent message tuple $\mathbf{i}$ on every nonzero entry, and have at least a $\gamma$ fraction of nonzero entries.
We will assume that each message in $[M_{1}]\times \cdots \times [M_{t}]$ is transmitted with equal probability, so that the average probability of error for a given adversarial choice of $\s$ is:
\begin{equation}
\label{eqn:avg_error_general}
e_{\gamma}(\s) = \frac{1}{M_{1}\cdots M_{t}}\sum_{\mathbf{i}}e_{\gamma}(\mathbf{i},\s).
\end{equation} 
We say that a rate tuple $(R_{1},\ldots,R_{t})\in \mathbb{R}_{\geq 0}^{t}$ is \textit{achievable for $\gamma$ partial correction} if there exists a sequence of $(2^{R_{1}n},\ldots,2^{R_{t}n},n)$ codes such that
$ \max_{\s}e_{\gamma}(\s)$ approaches $0$ with increasing block length $n$. As in a point-to-point AVC or the two-user AV-MAC case, $\text{argmax}_{\s}e_{\gamma}(\s)$ is the adversary's best chance of inducing a decoding error. 

The ($t$-dimensional) \textit{authentication capacity region} $\mathscr{C}_{\text{auth}}$ and the \textit{$\gamma$ partially correcting authentication capacity region}, $\mathscr{C}_{\text{auth},\gamma}$, are the closures of the sets of achievable rate tuples for each respective goal, where the former is realized when $\gamma = 0$. Let $\mathscr{C}$ denote the communication capacity region in the no-adversary setting (i.e., $\s=\s_{0}$, $\gamma =1$). We say that a capacity region has \textit{nonempty interior} if it contains a point such that all coordinate values are positive.

Critical to authentication and partial correction are the concepts of \textit{symmetrizability} \cite{CN88} and \textit{overwritability} \cite{KK18}: channel conditions which determine whether a channel is amenable to these types of communication. Below, we give extensions to the original point-to-point definitions to a $t$-user AV-MAC:

\begin{defn}
\label{defn:symm-MAC}
Let $t\geq 2$ and $m\in[t]$. A $t$-user AV-MAC $W_{Y|X_{1}\cdots X_{t}S}$ (denoted by $W$) is \textit{$X_{i_{1}}\times \cdots\times X_{i_{m}}$-symmetrizable} if there exists $P:=P_{S|X_{i_{1}}\cdots X_{i_{m}}}$ such that for all $x_{i_{1}},\ldots,x_{i_{m}}, x_{i_1}',\ldots,x_{i_{m}}', y$,
\begin{align*}
 \sum_{s}P(s\mid x_{i_1}',\ldots,&x_{i_m}')W(y|x_{i_1},\ldots,x_{i_m},s)= \\
  \sum_{s}P&(s\mid x_{i_1},\ldots,x_{i_m})W(y|x_{i_1}',\ldots,x_{i_m}',s).   
\end{align*} 
\end{defn}

The case $t=2$ results in the symmetrizability conditions of 
\cite{G90}, which along with \cite{J81,AC99} showed that (lack of) symmetrizability completely characterizes when the AV-MAC communication capacity region $\mathscr{C}$ has (non)empty interior.

\begin{defn}
\label{defn:overwrit-MAC}
Let $t\geq 2$ and $m\in [t]$. A $t$-user AV-MAC $W_{Y|X_{1}\cdots X_{t} S}$ (denoted by $W$) is \textit{$X_{i_{1}}\times \cdots\times X_{i_{m}}$-overwritable} if there exists $P:=P_{S|X_{i_{1}}\cdots X_{i_{m}}}$ such that for all $x_{i_{1}},\ldots,x_{i_{m}}, x_{i_1}',\ldots,x_{i_{m}}', y$,
\begin{align*}
 \sum_{s}P(s\mid x_{i_1}',\ldots,x_{i_m}')W(y|x_{i_1},& \ldots,x_{i_m},s)= \\ 
 & W(y|x_{i_1}',\ldots,x_{i_m}',s_0).   
\end{align*} 

\end{defn}

Again, the case of $t=2$ reduces to previous results: it was shown in \cite{BGKKY20} that (lack of) overwritability completely characterizes when the authentication capacity region $\mathscr{C}_{\text{auth}}$ has (non)empty interior.

{
For brevity, we will say that a channel is \textit{$m$-symmetrizable} (resp., \textit{-overwritable}) if there exists some subset of $m$ users $i_{1},\ldots,i_{m}$ such that the channel is $X_{i_{1}}\times \cdots\times X_{i_{m}}$-symmetrizable (resp., overwritable).}


\section{Necessary conditions for nonempty interior}
\label{sec:cap_region}

Previous work completely classified the authentication capacity region $\mathscr{C}_{\text{auth}}$ for the case of two users, and established necessary conditions for nonempty interior of the $\gamma=0.5$ partially correcting authentication capacity region $\mathscr{C}_{\text{auth},0.5}$ in the same setting \cite{BGKKY20}. In this section, we extend these results to more than two users. Because the authentication rate region is not the primary topic of this paper, and the results extend in a straightforward way to more users, we omit the following proof pertaining to $\ScrC_\auth$; 
this result extends Lemma III.6 and Theorem III.7 of \cite{BGKKY20}.

\begin{thm}
\label{thm:auth_empty}
    A $t$-user AV-MAC is $m$-overwritable for some $m\in [t]$ if and only if $\mathscr{C}_{\text{auth}}$ has empty interior. Otherwise, $\mathscr{C}_{\text{auth}}=\mathscr{C}$.
\end{thm}

The following result on $\ScrC_{\auth, \gamma}$ can be seen by observing that any $\gamma$ partially correcting authentication code is simultaneously an authentication code. The result extends Lemma IV.2. of \cite{BGKKY20}.

\begin{lem}
\label{lem:gamma_contained_auth}
For any $t$-user AV-MAC and $\gamma \in (0,1)$, $\mathscr{C}_{\text{auth},\gamma} \subseteq \mathscr{C}_{\text{auth}}$.
\end{lem}

Theorem \ref{thm:auth_empty} and Lemma \ref{lem:gamma_contained_auth} together imply that non-$m$-overwritability (for all $m\in[t]$) is a necessary condition for $\mathscr{C}_{\text{auth},\gamma}$ to have nonempty interior.
Next, we give another necessary condition for nonempty interior of $\mathscr{C}_{\text{auth},\gamma}$. Namely, the channel can only be symmetrizable up to the number of users we need \textit{not} correct to achieve $\gamma$ partial correction. While this result extends Theorem IV.3 of \cite{BGKKY20}, its proof contains more subtlety than the two-user case, so we include it here in full.

\begin{thm}
\label{thm:gamma_empty_int}
Let $\gamma\in (0,1)$. If a $t$-user AV-MAC $W_{Y|X_{1}\cdots X_{t}S}$ is $m$-symmetrizable for any $m > t-\lceil \gamma t\rceil $, then $\mathscr{C}_{\text{auth},\gamma}$ has empty interior. 
\end{thm}

\begin{proof}
Let $\gamma \in (0,1)$. Suppose $W:=W_{Y|X_{1}\cdots X_{t}S}$ is $X_{i_1}\times \cdots \times X_{i_{m}}$-symmetrizable, where $m > t-\lceil \gamma t\rceil $. 
Without loss of generality, let $i_j=j$, so that the coordinates in question are the first $m$,  and let $P:=P_{S|X_{1}\cdots X_{m}}$ be an adversarial distribution satisfying the property of Definition \ref{defn:symm-MAC}. 
Consider a sequence of $(M_{1},\ldots,M_{t},n)$ codes, with $M_{j}:=2^{R_{j}n}$ where $R_j>0$ for $j\in [t]$.
Let $\mathbf{x}(\mathbf{i})$ be the encoding of message vector $\mathbf{i}$, and let $\mathbf{v}_a^b$ denote coordinates $a$ through $b$ of a vector $\mathbf{v}$. Define
$A_\mathbf{i}:=\{ \mathbf{\hat{i}} : \hat{i}_{j}\in\{0,i_{j}\} \text{ for } j\in[t], |\supp{\hat{\mathbf{i}} }|\geq \lceil\gamma t\rceil\}$, as in Section \ref{sec:prelims}. Finally, define $M:=(M_1\cdots M_m)(M_1\cdots M_t)$.
Then, $\max_{\s}e_{\gamma}(\s)$ is bounded below by the expected value of $e_{\gamma}(\s)$ over $S$:
\begin{align}
&\geq \sum_{\s}\left(\frac{1}{M_1\cdots M_m}\sum_{\mathbf{i}_{1}^{m}}P(\s \mid \mathbf{x}(\mathbf{i})_{1}^{m})\right)e_{\gamma}(\s) \\
    &= \frac{1}{M}\sum_{\mathbf{i}_{1}^{m},\mathbf{k},\s}P(\s \mid \mathbf{x}(\mathbf{i})_{1}^{m})e_{\gamma}(\mathbf{k},\s) \label{eqn:def_error}
    \end{align}
    \begin{align}
        &\geq \frac{1}{M}\sum_{\mathbf{i}_{1}^{m},\mathbf{k},\s}P(\s \mid \mathbf{x}(\mathbf{i})_{1}^{m})W(\phi^{-1}(A_\mathbf{k})^{c}\mid \x(\mathbf{k}), \s) \label{eqn:def_channel}\\
            &= \frac{1}{M}\sum_{\mathbf{i}_{1}^{m},\mathbf{k},\s}P(\s \mid \mathbf{x}(\mathbf{k})_{1}^{m})W(\phi^{-1}(A_\mathbf{k})^{c}\mid \x(\mathbf{i}_1^m\mathbf{k}_{m+1}^{t}), \s) \label{eqn:symm_bad}
\end{align}
where Equations \eqref{eqn:def_error} and \eqref{eqn:def_channel} follow by definition, and \eqref{eqn:symm_bad} from symmetrizability. 

Now, we consider the sets $A_{\mathbf{i}_1^m\mathbf{k}_{m+1}^{t}}$ and $A_\mathbf{k}$. If $i_j\neq k_j$ for all $j\in [m]$, then these two sets are disjoint: indeed, any decoded message tuple with support of size at least $\lceil \gamma t\rceil$ must contain at least one coordinate from the first $m$ (recall that $m > t-\lceil \gamma t\rceil$). In other words, $\phi^{-1}(A_{\mathbf{i}_1^m\mathbf{k}_{m+1}^{t}})\subseteq \phi^{-1}(A_\mathbf{k})^{c}$ when $i_j\neq k_j$ for all $j\in [m]$. Using that $R_{j}>0$, and thus that the set of $\mathbf{i}_{1}^{m}$'s such that $i_j\neq k_j$ for fixed $\mathbf{k}$ is nonempty,
\begin{align}
&\geq   \frac{1}{M}\sum_{\substack{\mathbf{k},\s\\\mathbf{i}_{1}^{m}: {i}_{j}\neq k_{j}}} P(\s \mid \mathbf{x}(\mathbf{k})_{1}^{m})\cdot \\
& \quad \quad \quad\quad \quad\quad \quad W(\phi^{-1}(A_{\mathbf{i}_1^m\mathbf{k}_{m+1}^{t}})\mid \x(\mathbf{i}_1^m\mathbf{k}_{m+1}^{t}), \s)\\
    &=  \frac{1}{M}\sum_{\substack{\mathbf{k},\s\\\mathbf{i}_{1}^{m}: {i}_{j}\neq k_{j}}}P(\s \mid \mathbf{x}(\mathbf{k})_{1}^{m})\left(1-e_{\gamma}(\mathbf{i}_1^m\mathbf{k}_{m+1}^{t},\mathbf{s})\right) \\
        &=  \frac{1}{M}\sum_{\substack{\s, \mathbf{k}_{1}^{m}}}P(\s \mid \mathbf{x}(\mathbf{k})_{1}^{m})\sum_{\substack{\mathbf{i}_{1}^{m}\mathbf{k}_{m+1}^{t}:\\
        i_j\neq k_j}}\left(1-e_{\gamma}(\mathbf{i}_1^m\mathbf{k}_{m+1}^{t},\mathbf{s})\right) \\
            &\geq \frac{1}{M}\sum_{\s, \mathbf{k}_{1}^{m}}P(\s \mid \mathbf{x}(\mathbf{k})_{1}^{m})\left({\prod_{a=1}^{m}(M_a-1)\prod_{b=m+1}^{t}M_{b}}-e_{\gamma}(\s)\right)\label{eqn:use_pos_rate}\\
                &\geq \left(\frac{\prod_{a=1}^{m}(M_a-1)\prod_{b=m+1}^{t}M_{b}}{M_1\cdots M_t}-\max_{\s}e_{\gamma}(\s)\right)\frac{\sum_{ \mathbf{k}_{1}^{m}}1}{M_1\cdots M_m}\\
                    &=  \frac{\prod_{a=1}^{m}(M_a-1)\prod_{b=m+1}^{t}M_{b}}{M_1\cdots M_t}-\max_{\s}e_{\gamma}(\s)
\end{align}
Altogether,
\begin{equation}
    \max_{s}e_{\gamma}(\s)\geq  \frac{\prod_{a=1}^{m}(M_a-1)\prod_{b=m+1}^{t}M_{b}}{2M_1\cdots M_t}. 
\end{equation}
The lower bound approaches 0.5 in $n$ given that $R_j>0$ for $j\in [m]$, bounding $\max_{s}e_{\gamma}(\s)$ away from zero. We conclude that it is not possible that all $R_j$'s, $j\in [m]$, were positive. Thus, $\mathscr{C}_{\text{auth},\gamma}$ has empty interior.
\end{proof}

As in the two user case, it is possible that a channel is not overwritable in any sense, but is $m$-symmetrizable for some $m > t - \lceil \gamma t \rceil$; in this case, Theorem \ref{thm:gamma_empty_int}  tells us that $\mathscr{C}_{\text{auth},\gamma}$ has empty interior, {even while $\mathscr{C}_{\text{auth}}=\mathscr{C}$ may not. 
Furthermore, the proof of Theorem \ref{thm:gamma_empty_int} shows something slightly stronger than what is stated in the theorem:
 the {projection} of $\mathscr{C}_{\text{auth},\gamma}$ onto any $m$ symmetrizable coordinates ($m>t-\lceil \gamma t\rceil$) must have empty interior.


\section{Block Length Extension Scheme}
\label{sec:extension_scheme}
In this section, we present a method for extending the block length of a $\gamma$ partially correcting authentication code whose probability of $\gamma$ partial correction error is equal to zero. We note that unlike the classical ($\gamma = 1$) message correction case, simple concatenation of such a code will not automatically achieve the same rate as that of the original codebook: this is due to the fact that the particular $\lceil \gamma t \rceil$ users whose messages can be corrected in each time instance may vary. To adapt to our scenario, we use a concatenated code with the inner code tuple equal to a $\gamma$ partially correcting code with probability of $\gamma$ partial correction error equal to zero, and outer codes $C_{\varepsilon,j}$ given by a point-to-point codes designed for an induced erasure channel with a power constrained adversary. This induced channel is described in detail later in this section. A simplified version of this scheme appears in \cite{BGKKY20} for a particular two-user AV-MAC (a channel that is extended to more users in the case study of Section \ref{sec:case_study}). Here, we extend the scheme to the case of an arbitrary number of users.
We note that our scheme is not channel-dependent beyond the assumption that such an inner code exists.

 Suppose we have a set of $t$ codebooks, each of block length $n$, that have correction capability $\gamma:=\frac{u}{t}$ (with zero probability of $\gamma$ partial correction error).\footnote{For fixed $t$ and integer $1\leq u<t$, we take $\frac{u-1}{t}\leq \gamma<\frac{u}{t}$ and ``round'' it to $\frac{u}{t}$. This will not affect the number of users correctable due to the ceiling function on $\lceil \gamma t\rceil$.} If each user concatenates $r$ codewords from their codebook, at least $u$ users are correctable in any given time instance, while the remainder will be deemed to be in erasure. Notice that at most $t-u$ users experience erasure in each of the $r$ time instances.
Our outer codes, $C_{\varepsilon,j}$, must protect against these erasures for at least $u$ of the users. To gain an understanding of how this induced erasure channel functions, consider the following example:

\begin{example}
    \label{example: extending blocklength}
    Consider a set of three codebooks that have partial correction capability $\gamma=\frac{2}{3}$ and block length $n$, where 
        $C_i:=f_i(M_i)=\{c_{i1},c_{i2}\} $
    for $i\in[3]$.
    To extend the block length, we will use outer codes $C_{\varepsilon,j}\subseteq \mathbb{F}_{2}^{6}$. For example, let $C_{\varepsilon,j}=\{100101, 011101, 000101, 111010\}$ for all $j\in[3]$.
    Let 0 be replaced by the first codeword in each of the $C_j$'s and 1 be replaced by the second codeword in each of the $C_j$'s. For example the first codebook would become 
    \begin{align*}
        C_1' =\{c_{12} c_{11} & c_{11} c_{12} c_{11} c_{12}, \ c_{11} c_{12} c_{12} c_{12} c_{11} c_{12}, \\
        &  c_{11} c_{11} c_{11} c_{12} c_{11} c_{12}, \ c_{12} c_{12} c_{12} c_{11} c_{12} c_{11}\}
    \end{align*}
    with block length $6n$ and rate $ \frac{2}{6n} = \frac{1}{3n}$. 
    Because the $C_j$'s can correct $\gamma t =2$ of the three users, there will be a maximum of one erasure in each time instance.
    An example erasure pattern is given below, where an erasure is represented by $\varepsilon$. 
    Note that with this erasure pattern, we can still correct two of the three users' messages, and thus achieve the goal of $\gamma$ partial correction.

     \begin{center}
    \begin{tabular}{ c c c c c c c }
     \textbf{User 1} & \underline{ $\varepsilon$ } & \underline{ $\varepsilon$ } & \underline{ $\varepsilon$ } & \underline{ $\varepsilon$ }  & \underline{ $\varepsilon$ } & \underline{ $\varepsilon$ }\\ 
     
     \textbf{User 2} & \underline{$c_{22}$} & \underline{$c_{21}$} & \underline{$c_{21}$} & \underline{$c_{22}$} & \underline{$c_{21}$} & \underline{$c_{22}$}  \\  
     
     \textbf{User 3} & \underline{$c_{31}$} & \underline{$c_{32}$} & \underline{$c_{32}$} & \underline{$c_{32}$} & \underline{$c_{31}$} & \underline{$c_{31}$}
     
    \end{tabular}
    \end{center}

\end{example}

The above erasure pattern example suggests that the adversary's best strategy will be to spread their efforts across enough users, but not any more than needed, in order to deter $\gamma$ partial correction. Intuitively, the adversary should choose to target $t-u+1$ users to have the most erasures per affected user while not leaving $u$ users with zero erasures. 
The optimal strategy is formalized in the proof of the following lemma.

\begin{lem}
\label{lem:max_erasures}
Let $t\geq 2$ and $1\leq u <t$. If a $\gamma = \frac{u}{t}$ partially correcting codebook tuple (with error probability zero) is concatenated $r$ times, at least $u$ users will experience at most $\lfloor\frac{r(t-u)}{t-u+1}\rfloor$ total erasures.
\end{lem}

\begin{proof}
In each time instant, the adversary can attempt to erase $t-u$ users' symbols. If they are always successful, there are a total of $r(t-u)$ erasures across all users and all time instances. Suppose that the adversary has full control over which users will experience erasures, and they choose to restrict these erasures to $\Gamma$ of the $t$ users. 
First, suppose $\Gamma < t-u+1$. In this case, there are at least $u$ users that have zero erasures, and we are done.
Now, let $\Gamma \geq t-u+1$. The average number of erasures per targeted user is $\frac{r(t-u)}{\Gamma}$.

We claim that at least $\Gamma-(t-u)$ users have at most $\lfloor\frac{r(t-u)}{t-u+1}\rfloor$ erasures. Suppose not, and that at least $\Gamma-[\Gamma-(t-u)]+1=t-u+1$ users have strictly more than $\lfloor\frac{r(t-u)}{t-u+1}\rfloor$ erasures. If $\frac{r(t-u)}{t-u+1}$ is an integer, the total number of erasures across all users and time instances is strictly bounded below by
$(t-u+1)\frac{r(t-u)}{t-u+1} = r(t-u)$.
If it is not an integer, the total number of erasures would be bounded below by $(t-u+1)\lceil \frac{r(t-u)}{t-u+1}\rceil  > r(t-u)$.
Both cases contradict that the total number of erasures is equal to (at most) $r(t-u)$. 

Thus, at least $\Gamma-(t-u)$ users have at most $\lfloor\frac{r(t-u)}{t-u+1}\rfloor$ erasures.
The $t-\Gamma$ non-targeted users have zero erasures. Thus, at least $t-\Gamma+(\Gamma-(t-u))=u$
users have at most $\lfloor\frac{r(t-u)}{t-u+1}\rfloor$ erasures.
\end{proof}

This upper bound on the number of erasures for some subset of $u= \lceil \gamma t\rceil $ users is tight if the adversary may choose which users to target, and if they are able to reliably erase their targeted users. Both are advantageous assumptions for the adversary; we note that they will not always be the case (see Section \ref{sec:case_study} for a channel case study without the latter property).

\begin{example}
     Consider the codebooks of Example \ref{example: extending blocklength} with $\gamma = \frac{u}{t}=\frac{2}{3}$.
    First consider the case were the adversary targets all users equally. A possible erasure pattern is given below:
    \begin{center}
    \begin{tabular}{ c c c c c c c }
     \textbf{User 1} & \underline{ $\varepsilon$ } & \underline{ $\varepsilon$ } & \underline{$c_{11}$} & \underline{$c_{12}$}  & \underline{$c_{11}$} & \underline{$c_{12}$}\\ 
     
     \textbf{User 2} & \underline{$c_{22}$} & \underline{$c_{21}$} & \underline{ $\varepsilon$ } & \underline{ $\varepsilon$ } & \underline{$c_{21}$} & \underline{$c_{22}$}  \\  
     
     \textbf{User 3} & \underline{$c_{31}$} & \underline{$c_{32}$} & \underline{$c_{32}$} & \underline{$c_{32}$} & \underline{ $\varepsilon$ } & \underline{ $\varepsilon$ }
     
    \end{tabular}
    \end{center}
    In this case, the decoder needs to be able to correct two erasures (per user) in order to correct two of the three users. 

    Next, we look at the case were the adversary focuses their efforts on $t-u+1=2$ users. Per Lemma \ref{lem:max_erasures}, this is the adversary's optimal strategy.

     \begin{center}
    \begin{tabular}{ c c c c c c c }
     \textbf{User 1} & \underline{ $\varepsilon$ } & \underline{ $\varepsilon$ } & \underline{ $\varepsilon$ } & \underline{$c_{12}$}  & \underline{$c_{11}$} & \underline{$c_{12}$}\\ 
     
     \textbf{User 2} & \underline{$c_{22}$} & \underline{$c_{21}$} & \underline{$c_{21}$} & \underline{ $\varepsilon$ } & \underline{ $\varepsilon$ } & \underline{ $\varepsilon$ }  \\  
     
     \textbf{User 3} & \underline{$c_{31}$} & \underline{$c_{32}$} & \underline{$c_{32}$} & \underline{$c_{32}$} & \underline{$c_{31}$} & \underline{$c_{32}$}
     
    \end{tabular}
    \end{center}
    Here, the decoder needs to correct three erasures in order to correct two of the three users.
\end{example}

\begin{remark}
\label{rem:min_dist}
According to Lemma \ref{lem:max_erasures}, if we wish to design $C_{\varepsilon,j}$ with probability of $\gamma$ partial correction error equal to zero,
$d_{min}(C_{\varepsilon,j})\geq \lfloor \frac{r(t-u)}{t-u+1} \rfloor +1$ for each $j\in[t]$. 
\end{remark}

Remark \ref{rem:min_dist} addresses the requirement of perfect correction of $C_{\varepsilon,j}$. Allowing for some vanishing decoding error probability, we turn to the capacity of the induced erasure AVC. 

\begin{lem}
\label{lem:extension}
Let $\gamma\in(0,1)$, $n\geq 1$, $t\geq 2$, and $W:=W_{Y|X_1\cdots X_tS}$ be a $t$-user AV-MAC. Suppose a $\gamma$ partially correcting authentication code $(M_1,\ldots,M_t,n)$ exists for $W$ such that $M_j:=2^{R_{j}n}>1$ for all $j\in [t]$ and the probability of $\gamma$ partial correction error is equal to zero. Then, $\mathscr{C}_{\text{auth},\gamma}$ has nonempty interior.
\end{lem}

\begin{proof}
Choose a user $j\in[t]$, and 
define a deterministic erasure AVC as follows: let $\mathcal{X}=[M_j]$, $\mathcal{Y}=[M_j]\cup \{\varepsilon\}$, and $\mathcal{S}=\{s_{0},s_1\}$. Then, $y=x$ if $s=s_0$, and $y=\varepsilon$ if $s=s_1$.
Let $u:=\lceil \gamma t \rceil $. Importantly, the adversary is power-constrained so that in a length-$r$ transmission they may choose at most $\lfloor \frac{r(t-u)}{t-u+1}\rfloor$ coordinates to be equal to $s_1$; the remainder must be equal to $s_0$. There are no constraints on the legitimate user's channel input. This channel mimics the worst-case scenario for (at least) $u$ users in the above-described concatenation scheme: each has at most $\lfloor\frac{r(t-u)}{t-u+1}\rfloor$ erasures, and in our induced channel we assume that any time the adversary attempts to erase a user they can do so successfully. 
We refer the reader to Theorem 3 of \cite{CN88} (and the forthcoming full version of this work) to verify that this channel has positive capacity.

It remains to explain why this implies an achievable positive rate tuple for the original AV-MAC. For each user $j\in [t]$, let a code sequence $(2^{Q_{j}r},r)$ achieve the capacity of the erasure AVC. Using the concatenation scheme described earlier in this section, with the existing zero-error $\gamma$ partial correction code as inner code, we achieve rate $Q_{j}R_{j}>0$ for user $j$.
\end{proof}

The case study in \cite{BGKKY20} calculates the exact capacity of an induced erasure AVC, which is dependent on the AV-MAC studied there (and extended in Section \ref{sec:case_study}), as well as the specific choice of inner code for that channel. There, when the adversary acted they had a 0.5 probability of erasing, as opposed to the guaranteed erasure assumed in the proof of Lemma \ref{lem:extension}. In other words, we believe it is possible to be more specific about the values of the positive rate tuples achievable using our extension scheme. We plan to address this question, as well as the question of whether a less stringent inner code may be utilized, in a full version of this work.


\section{Zero Probability of Partial Correction Error Codes Case Study}
\label{sec:case_study}

In this section, we turn to a particular $t$-user channel satisfying the necessary conditions of Section \ref{sec:cap_region} for authenticated partial correction. We will work to construct short block length codes with zero probability of $\gamma$ partial correction error, with the knowledge that such codes can be extended using the scheme of Section \ref{sec:extension_scheme}.
To define the channel, let $\mathcal{X}_1=\cdots=\mathcal{X}_t=\{0,1\}$, $\mathcal{S}=\{0,1,2,\ldots,\ell\}$ for some $\ell \geq 1$, and $\mathcal{Y}=\{0,1,\ldots, t+\ell\}$, where $Y=X_1+X_2+\cdots +X_t +S$. The no-adversary state is given by $s_0=0$. For ease of notation, we will denote this channel by $W^{+}_{t,\ell}$.
Observe that $W^{+}_{t,\ell}$ is deterministic given a choice of state. 

\subsection{Necessary conditions are satisfied}

We first verify that $W^{+}_{t,\ell}$ satisfies the necessary overwritability and symmetrizability conditions established in Section \ref{sec:cap_region}.

\begin{lem}
\label{lem:W+_over}
$W^{+}_{t,\ell}$ is not $m$-overwritable for any $m\in [t]$. 
\end{lem}

\begin{proof}
Let $W:=W^{+}_{t,\ell}$ and $m\in[t]$. 
Toward contradiction, assume the channel is $m$-user overwritable in the first $m$ coordinates, and let $P$ be the distribution guaranteed by the definition of overwritability. Let $x_1'=\cdots =x_m'=1$, $x_1=\cdots =x_t=0$, and $y=0$. Then, 
\begin{align}
  \sum_{s=0}^{\ell} P(s\vert \underbrace{1,\ldots,1}_{m})W(0&\vert \underbrace{0,\ldots, 0}_{t},s)=\nonumber \\
  &W(0\vert \underbrace{\overbrace{1,\ldots,1}^{m},\overbrace{0,\ldots,0}^{t-m}}_{t},0)  
\end{align}
On the left hand side, $W(0\vert 0,\ldots,0,s)=1$ if and only if $s=0$; otherwise it is zero. Therefore, the left side is equal to $P(0\vert 1,\ldots,1)$. On the right side, $W(0\vert 1,\ldots,1,0,\ldots,0,0)=0$. Thus, $P(0\vert 1,\ldots,1)=0$.

A similar argument with $x_1^{'}=\cdots =x_m^{'}=1$, $x_1=x_2=\cdots=x_m=1$, $x_{m+1}=\cdots=x_t=0$, and $y=m$ yields $P(0|1,\dots,1)=1$. 
This is a contradiction because $P(0\vert 1,\ldots,1)$ cannot be equal to both 0 and 1. Therefore, the channel is not $m$-user overwritable.
\end{proof}

Recall that the other necessary condition for a $t$-user channel to be $\gamma$ partially correcting is that the channel is not $q$-user symmetrizable for any $q\geq t-\lceil\gamma t\rceil$. The following establishes allowed values of $\gamma$ given the adversary's power constraint $\ell$.

\begin{lem}
\label{lem:W+_symm}
    Let $\ell \leq t$. Then $W^{+}_{t,\ell}$ is $q$-user symmetrizable for any subset of $q$ users exactly when $q\leq \ell$.
\end{lem}

\begin{proof}
Let $W:=W^{+}_{t,\ell}$. First we show that the channel is $q$-user symmetrizable for $q\leq \ell$.
Consider the following probability distribution:
$P(s\vert x_1, ..., x_q)= 1$ if $\sum_{i=1}^{q} x_i=s$, and $0$ otherwise.
Notice that because each $x_i\in\{0,1\}$, we have $0\leq  \sum_{i=1}^{q}x_i \leq q$. Because $q\leq \ell$, for a fixed choice of $x_i$'s, $P(s\vert x_1,\ldots, x_q)=1$ for exactly one choice of $s$.
Using this distribution in Definition \ref{defn:symm-MAC}, the left hand size yields
$
  W\left(y \vert  x_1,\ldots, x_t,  \sum_{i=1}^{q}x_i'\right)$, and the right hand size becomes
  $
 W\left(y \vert x_1',\ldots,x_q', x_{q+1},\ldots,x_t, \sum_{i=1}^{q}x_i\right)
$.
The sum of the inputs will be the same on both sides, so either both sides are equal to 1 if the sum of the inputs is equal to $y$ or both sides are equal to 0 if the sum of the inputs is not equal to $y$. Therefore, 
the channel is $q$-user symmetrizable for the first $q$ users. Notice that permuting the coordinates will not change the argument, so we have shown that the channel is $q$-user symmetrizable for any subset of $q$ users.

Now we will show the channel is \textit{not} $q$-user symmetrizable when $\ell+1\leq q \leq t$. By way of contradiction, assume the channel is $q$-user symmetrizable for such a $q$. Let $P$ be the distribution guaranteed by Definition \ref{defn:symm-MAC}, and let $z\in S$ (note that $z\leq \ell <t$). We claim that $P(z\vert0,...,0)=0$. Assume $x_1'=\ldots=x_q'=0$, $x_1=\ldots=x_{t-z}=1$, $x_{t-z+1}=\ldots=x_t=0$, and $y=t$. Then, the left hand side of the definition is the sum over $s\in \mathcal{S}$ of
\begin{equation}
\label{eq:left}
    P(s\vert\underbrace{ 0,\ldots,0}_{q})W(t\vert\underbrace{ \overbrace{1,\ldots, 1}^{t-z},{0, \ldots,0}}_{t},s)  
\end{equation}
and the right hand side is equal to the sum over $s\in\mathcal{S}$ of
\begin{equation}
\label{eq:right}
    P(s\vert \underbrace{\overbrace{1,\ldots,1}^{\min\{t-z,q\}},{0,\ldots,0}}_{q})W(t\vert\underbrace{\overbrace{ 0,\ldots, 0}^{q},\overbrace{1,\ldots,1}^{((t-z)-q)^+},{0,\ldots,0}}_{t},s)
\end{equation}
where we use $\lambda^{+}$ to denote $\max\{\lambda,0\}$.
In Equation \eqref{eq:left}, $W(t\vert1,\ldots,1,0,\ldots,0,s)=1$ if and only if $s=z$. Therefore, \eqref{eq:left} is equal to $P(z\vert 0,\ldots,0)$. 
In Equation \eqref{eq:right}, $W(t|0,\dots,0,1,\dots,1,0,\ldots,0,s)=1$ if and only if $t=(t-z-q)^++s$. Since $s\leq \ell<q\leq t$, it follows that $s-q<0$, and also $s<t$. Thus, $(t-z-q)^++s<t$, and \eqref{eq:right} equals 0.
Therefore, $P(z\vert 0,\ldots,0)=0$ for all $0\leq z\leq \ell$,
a contradiction. Therefore, no such $P$ exists and the channel is not $q$-user symmetrizable for $\ell+1\leq q\leq t$.
\end{proof}

Lemmas \ref{lem:W+_over} and \ref{lem:W+_symm} together establish the following:

\begin{thm}
    The channel $W_{t,\ell}^{+}$ satisfies the necessary conditions for $\gamma$ partial correction established by Theorem \ref{thm:auth_empty}, Lemma \ref{lem:gamma_contained_auth}, and Theorem \ref{thm:gamma_empty_int}. 
\end{thm}

\subsection{Zero probability of $\gamma$ partial correction characterization}

To aid in our discussion of code construction, we next introduce notation that will help explain when a codebook is $\gamma$ partially correcting with zero probability of error for $W_{t,\ell}^{+}$.
Let $C_j := f_{j}(M_{j})$ denote the block length $n$ codebook of user $j$, so that each $\mathbf{x}\in C_{j}$ is equal to $f_j(i)$ for some $i\in M_{j}$.  Then, define the following multisets:
 \begin{align}
     A_0 &:=\left\{\mathbf{u}=\sum_{j=1}^{t}\mathbf{x}_j \ \biggr\vert \ \mathbf{x}_j\in C_j \right\} \\ 
     A_1&:=\left\{\mathbf{u}=\mathbf{s}+\sum_{j=1}^{t}\mathbf{x}_j \ \biggr\vert \ \mathbf{x}_j\in C_j, \ \mathbf{s}\in \mathcal{S}\setminus \{\mathbf{s}_{0}\} \right\}
 \end{align}
In other words, the set $A_{0}$ is the set of all possible channel outputs (with multiplicity) when the adversary does not act, and $A_{1}$ is the set of all outputs when the adversary does act.

\begin{lem}
\label{lem:A_char}
    Over the channel $W_{t,\ell}^{+}$, a codebook $t$-tuple is $\gamma$ partially correcting with zero probability of $\gamma$ partial correction error if and only if all of the following hold:
    \begin{enumerate}[(1)]
        \item $A_0\cap A_1=\emptyset$;
        \item The elements of $A_0$ are unique;
        \item For each $\mathbf{w}\in A_1$ that appears with multiplicity, there exists some subset $J=\{j_1,j_2,\ldots,j_{\lceil \gamma t\rceil }\}\subseteq [t]$ such that if $\mathbf{s}+\sum_{j=1}^{t}\mathbf{x}_j=\mathbf{w}$ and $\mathbf{s}'+\sum_{j=1}^{t}\mathbf{x}'_j=\mathbf{w}$, then $\mathbf{x}_i=\mathbf{x}'_i$ for all $i\in J$.
\end{enumerate}
\end{lem}

\begin{proof}
    Condition (1) ensures that the decoder can reliably distinguish between the case where the adversary has acted and the case where they have transmitted a sequence of all zeros ($\mathbf{s}_0$). Taken together with (2), we have $e_{\gamma}(\mathbf{i},\mathbf{s}_{0})=0$ for every message tuple $\mathbf{i}$. If either condition fails, $e_{\gamma}(\mathbf{i},\mathbf{s}_{0})>0$.

    With (1), condition (3) establishes that $e_{\gamma}(\mathbf{i},\mathbf{s})=0$ when $\mathbf{s}\neq \mathbf{s}_{0}$: if there are repeated elements in $A_1$, we are guaranteed to be able to correct $\lceil \gamma t\rceil$ of the messages, even if the others must be discarded. If condition (3) fails, $e_{\gamma}(\mathbf{i},\mathbf{s})>0$.
\end{proof}
With this characterization in hand, we turn to short codebook design strategies. For the remainder of the paper, we will focus on the codebooks $C_j:=f_{j}(M_j)$; thus, we will will discuss codebook tuples of the form $(C_1,\ldots,C_t)$.

\subsection{Two users}

Here we present necessary conditions for $\gamma = 0.5$ partially correcting codebook pairs over $W_{2,1}^{+}$ with zero probability of $\gamma$ partial correction error, and bound the sizes of these codebooks for fixed block length.

\begin{example}
     \label{ex:BeemerAVMAC}
    The codebook pair $C_1=\{011, 100\}$, and $C_2=\{010, 101\}$ is $0.5$ partially correcting over $W_{2,1}^{+}$ with zero probability of $\gamma$ partial correction error. This example was explored extensively in \cite{BGKKY20}.
\end{example}

The each codebook of Example \ref{ex:BeemerAVMAC} has the property that codeword supports {are not contained in one another.}
We find that this is true in general for partial correction over $W_{t,\ell}^{+}$.

\begin{thm}
    \label{thm:vec_supps}
    Let $(C_1,\dots,C_t)$ be a codebook tuple with $C_j\in \{0,1\}^{n}$ for $j\in[t]$. If, for some $j\in[t]$, $\x,\y\in C_j$ with $\x\neq \y$ and $\text{\em supp}(\x)\subseteq\text{\em supp}(\y)$, then the codebook tuple $(C_1,\dots,C_t)$ is not $\gamma$ partially correcting with zero probability of $\gamma$ partial correction error over $W_{t,\ell}^{+}$ for any $\gamma\in(0,1)$.
\end{thm}

\begin{proof}
    Let $\gamma \in (0,1)$. Suppose that for some $j\in[t]$ there exist $\x\neq \y\in C_j$ such that $\supp{x}\subseteq\supp{y}$. Let $\supp{y}\setminus\supp{x}$ be the support of the adversarial contribution $\s$. Notice that $\s\in \{0,1\}^{n}\subseteq \mathcal{S}^{n}$, and $\s+\mathbf{x}=\mathbf{y}+\mathbf{0}$. 
    Thus $A_0\cap A_1\neq \emptyset$. By Lemma \ref{lem:A_char}, the result follows.
\end{proof}

Let the partially ordered set (poset) $\mathcal{P}([n])$ be the power set of $[n]$ together with the partial order defined by set inclusion. Elements of the poset can alternatively be thought of as vectors in $\{0,1\}^n$; elements whose supports are contained in one another are then {related} under the partial order.
Theorem \ref{thm:vec_supps} states that each codebook of a $\gamma$ partially correcting (with zero probability of $\gamma$ partial correction error) codebook tuples are antichains (sets of unrelated elements) in $\mathcal{P}([n])$. The following corollary is a direct consequence of Theorem \ref{thm:vec_supps}.

\begin{cor}
    Suppose the codebook tuple $(C_1,\dots,C_t)$ is $\gamma$ partially correcting with zero probability of $\gamma$ partial correction error over $W_{t,\ell}^{+}$ for some $\gamma\in(0,1)$. If $|C_j|>1$, then $\mathbf{0},\mathbf{1}\notin C_j$.
\end{cor}
While each individual codebook must be an antichain, the disjoint union of codebooks need not be. Indeed, Example \ref{ex:BeemerAVMAC} has two (disjoint) related pairs across codebooks. The following places a limit on such support containment.
\begin{thm}
\label{thm:potpourri}
Let $(C_1,C_2)$ be a $0.5$ partially correcting codebook pair with zero probability of $\gamma$ partial correction error over $W_{2,1}^{+}$. Provided $|C_1|,|C_2|>1$, $C_1\cap C_2=\emptyset$ and $C_1\sqcup C_2$ has at most two pairs of related codewords having the property that the intersection of these pairs is empty.
\end{thm}
\begin{proof}
Let $(C_1,C_2)$ be a $0.5$ partially correcting codebook pair with zero probability of $\gamma$ partial correction error over $W_{2,1}^{+}$ and $|C_1|,|C_2|>1$.
To show that the codebooks are disjoint, suppose $\x\in C_1\cap C_2$. Let $\mathbf{a}\neq \x\in C_1$ and $\mathbf{b}\neq \x\in C_2$. Note that $\mathbf{a},\mathbf{b},\x\in \mathcal{S}^{n}$. Observe that if either [$\x_1=\x$, $\x_2=\mathbf{b}$, $\s=\mathbf{a}$], or [$\x_1=\mathbf{a}$, $\x_2=\x$, $\s=\mathbf{b}$], the channel output is $\x+\mathbf{b}+\mathbf{a}$. This violates condition (3) of Lemma \ref{lem:A_char}. We have thus established that $C_1\cap C_2=\emptyset$.

Next, we will show that if there are two related pairs in the union, they have a particular structure.
Let $\x_1,\x_2,\y_1,\y_2\in C_1\sqcup C_2$ be distinct such that $\x_1\leq \x_2$ and $\y_1\leq \y_2$. Making use of Theorem \ref{thm:vec_supps}, and without loss of generality, we have two cases: either $\x_1,\y_1\in C_1$, or $\x_1,\y_2\in C_1$.
    
Suppose $\x_1,\y_1\in C_1$ and $\x_2,\y_2\in C_2$. Let $\mathbf{a},\mathbf{b}\in \mathcal{S}^{n}$ be such that $\x_1+\mathbf{a}=\x_2$ and $\y_1+\mathbf{b}=\y_2$.
Then,
\begin{equation}
    \x_1+\y_2+\mathbf{a}=\x_2+\y_2 =
    \y_1+\x_2+\mathbf{b},
\end{equation}
contradicting condition $(3)$ of Lemma \ref{lem:A_char}.
Thus, it must be the case that the two smaller elements, $\x_1$ and $\y_1$, must be in distinct codebooks.

Now suppose that $\x_1,\x_2,\y_1,\y_2,\z_1,\z_2\in C_1\sqcup C_2$ are distinct elements with $\x_1\leq \x_2, \ \y_1\leq\y_2, \ \z_1\leq\z_2$. Based on our above structural argument, $\x_1$, $\y_1$, and $\z_1$ must pairwise belong to different codebooks, an impossibility.
\end{proof}

Theorem \ref{thm:potpourri} implies an upper bound on the size of the disjoint union:

\begin{thm}
Let $(C_1,C_2)$ be a $0.5$ partially correcting codebook pair with zero probability of $\gamma$ partial correction error over $W_{2,1}^{+}$, with $|C_1|,|C_2|>1$. Then $|C_1\sqcup C_2|\leq \binom{n}{\lceil\nicefrac{n}{2}\rceil}+2$.
\end{thm}
\begin{proof}
    From Sperner's theorem, the size of a largest antichain in $\mathcal{P}([n])$ is $\binom{n}{\lceil\nicefrac{n}{2}\rceil}$. From Theorem \ref{thm:potpourri}, $C:=C_1\sqcup C_2$ will have the form of an antichain with at most two additional elements. The result follows. 
\end{proof}

\subsection{Three or more users}

With the goal of constructing short codebook tuples with zero probability of $\gamma$ partial correction error over $W_{t,\ell}^{+}$ for $t\geq 3$, in this section we develop equivalent conditions for (1) and (3) of Lemma
\ref{lem:A_char} which are easier to check computationally. In particular, we will reinterpret these conditions in terms of differences of sums of legitimate codewords. This allows us to avoid actually constructing $A_1$, and to instead check conditions on the elements of $A_0$ (i.e. sums of codewords).

In each of the results of this section, we will consider two such sums: $\mathbf{u}=\sum_{j=1}^{t}\x_{j}$ and $\mathbf{v}=\sum_{j=1}^{t}\y_{j}$, where (not necessarily distinct) $\x_j,\y_j\in C_j$ are vectors of length $n$ for each $j\in[t]$.

\begin{lem}
\label{lem:cond_1_rephrase}
   Let $(C_1,\ldots,C_t)$ be a codebook tuple for the channel $W_{t,\ell}^+$, where $\ell \geq 1$. There exist distinct $\mathbf{u,v}\in A_0$ such that $\mathbf{u}-\mathbf{v}$ is in $\{0,\ldots,\ell\}^{n}=\mathcal{S}^{n}$ if and only if  $A_0 \cap A_1 \neq \emptyset$.
\end{lem}
\begin{proof}
 Suppose the vector difference $\mathbf{u}-\mathbf{v}\in \{0,\ldots,\ell\}^n$ for some choice of $\mathbf{u,v}\in A_0$.  In this case,  $\mathbf{u}-\mathbf{v}$ is an element in $\mathcal{S}^n$ not equal to $\s_0$; call this difference $\s$. Then, $\u=\v+\s$. We then see that $ \u\in A_1$ and $\u \in A_0$. This means that $A_0 \cap A_1 \neq \emptyset$.
 On the other hand, let $\mathbf{u}$ be an element of nonempty $ A_{0}\cap A_{1}$. Since $\u\in A_1$, it must be the case that $\mathbf{u}=\mathbf{v}+\mathbf{s}$ for some $\mathbf{v}\in A_0$, $\s\neq \s_0$. Then, $\mathbf{u}-\mathbf{v}\in\mathcal{S}^{n}$, and we are done.
\end{proof}

The following two examples illustrate Lemma \ref{lem:cond_1_rephrase}.

\begin{example}
\label{ex:bad_ex_1}
  Consider the set of three codebooks with block length $n=6$ given below:
    \begin{align}
       C_1 &= \{100110, 110110\} \\
       C_2 &=\{111010, 100101\} \\
       C_3 &=\{011111, 001010\}
    \end{align}
    We claim that $A_0\cap A_1\neq\emptyset$ over $W_{3,\ell}^{+}$ for this codebook. Let $\u,\v\in A_0$ be given by  \begin{align}
       \u &= 100110+111010+011111 = 222231\\
       \v &= 110110+100101+001010=211221
    \end{align}    
    Then, $\u-\v=222231-211221=011010$.
    Here, we can see that $\u-\v$ is an element in $\mathcal{S}^n$ for any $\ell \geq 1$. Observe that $\u=\v+\s$ when $\s=011010$, so that $A_0 \cap A_1 \neq \emptyset$ and condition (1) of Lemma \ref{lem:A_char} fails.
\end{example}
\begin{example}
\label{ex:good_ex_1}
  Consider the set of three codebooks with block length $n=6$ given below:
    \begin{align}
       C_1 &= \{011010, 100101\} \\
       C_2 &=\{010110, 101001\} \\
       C_3 &=\{001101, 110010\}
    \end{align}
    We claim that $A_0\cap A_1=\emptyset$ over $W_{3,\ell}^{+}$ for this codebook. Consider the following choice of $\u$ and $\v$ as an example:
    \begin{align}
        \u &=011010 + 010110 + 001101= 022221 \\
        \v &=100101 + 101001 + 110010= 311112
    \end{align} Then, 
    \begin{align}
        \u-\v &=022221-311112=(-3)1111(-1)
    \end{align} Here we can see that $\u-\v$ is not an element in $\mathcal{S}^n$.
    All values of $\u-\v$ can be looped through for this channel to show that $A_0 \cap A_1 = \emptyset$.
\end{example}

Next, we rephrase condition (3) of Lemma \ref{lem:A_char} in terms of elements of $A_{0}$.

\begin{lem}
\label{lem:cond_3_rephrase}
 Let $(C_1,\ldots,C_t)$ be a codebook tuple for the channel $W_{t,\ell}^+$, where $\ell \geq 1$, and let $\gamma \in (0,1)$. Let $\mathbf{u,v}\in A_0$ such that $\u=\sum_{j=1}^{t}\x_j$ and $\v=\sum_{j=1}^{t}\y_j$ and $\x_j\neq \y_j$ for at least $t-\lceil \gamma t\rceil +1$ values of $j\in[t]$. If the maximum entry of $|\u-\v|$ is at most $\ell$, condition (3) of Lemma \ref{lem:A_char} fails. 
\end{lem}
\begin{proof}
Let $(C_1,\ldots,C_t)$ be a codebook tuple for the channel $W_{t,\ell}^+$, where $\ell \geq 1$, and let $\gamma \in(0,1)$. Let $\mathbf{u,v}\in A_0$ such that $\u=\sum_{j=1}^{t}\x_j$ and $\v=\sum_{j=1}^{t}\y_j$ and $\x_j\neq \y_j$ for at least $t-\lceil \gamma t\rceil +1$ values of $j\in[t]$.
   If the maximum entry of $|\u-\v|$ is $ \leq \ell$,  then $|\u-\v| \in \mathcal{S}^n$. Thus, $|\u-\v|=\s$ and $u_i-v_i=\pm s_i$ for each $i\in [n]$. Letting $\s_1$ equal $\s$ on coordinates where $u_i-v_i$ is negative, and zero elsewhere, and  $\s_2=\s$ when $u_i-v_i$ is positive, and zero elsewhere, we have $\u+\s_1 = \v+\s_2$. Condition (3) fails because the sums cannot match on any subset of $\lceil \gamma t\rceil$ codewords.
  
\end{proof}

The following example illustrates Lemma \ref{lem:cond_3_rephrase}. 
\begin{example}
  Consider the set of three codebooks with block length $n=6$ given below:    
  \begin{align}
        C_1 &= \{011010, 100101\}\\
        C_2 &=\{010110, 101001\}\\
        C_3 &=\{010101, 101010\}
    \end{align}
    We claim that condition (3) of Lemma \ref{lem:A_char} fails over $W_{3,\ell}^{+}$ when $\gamma = \frac{1}{3}$ or $\gamma =\frac{2}{3}$. Let $\u,\v\in A_0$ be 
    \begin{align}
        \u &= 100101+010110+101010=211221\\
        \v &= 011010+101001+010101=122112
    \end{align}
    It can be seen that $\u$ and $\v$ differ on all three users. 
    \begin{align}
        |\u-\v| &= |211221-122112|=111111
    \end{align}
    and $\u-\v=1(-1)(-1)11(-1)$
such that 
$\u+011001=\v+100110$.
Thus, condition (3) fails due to the fact that there is a repeated element in $A_{1}$ for which no user codewords match.
\end{example}

We observe that condition (2) of Lemma \ref{lem:A_char} is straightforward to check on $A_0$ alone. Combining all checks on $A_0$ described in this section,
we can algorithmically loop through the possible combinations of differences of elements of $A_0$ to test whether a codebook triple is a candidate for partial correction with zero probability of $\gamma$ partial correction error. Notably, the check of Lemma \ref{lem:cond_3_rephrase} is necessary (for some such $\u,\v$) for failure of condition (3) when there are three users and $\gamma=\frac{2}{3}$. In fact, the codebook given in Example \ref{ex:good_ex_1} is a good codebook triple for $W_{3,1}^{+}$ with $\gamma=\frac{2}{3}$. The extension scheme of Section \ref{sec:extension_scheme} can thus be used to achieve positive rate triples with arbitrary block length.
\section{Conclusion}
\label{sec:conclusion}
In this paper, we gave necessary (non-)symmetrizability and (non-)overwritability conditions for nonempty interior of the $\gamma$ partially correcting authentication capacity region over a $t$-user AV-MAC. We presented a scheme to extend the block length of a strong short block length code, showing that the resulting extension can maintain the positive rates of the short code. Finally, we examined the particular AV-MAC denoted $W_{t,\ell}^{+}$, deriving structural results and bounds for zero $\gamma$ partial correction error codes over this channel. Ongoing and future directions include sufficiency of the aforementioned necessary channel conditions for partial correction, refinement of our block length extension scheme, and alternative paths toward inner bounds on the $\gamma$ partially correcting authentication capacity region.

\section*{Acknowledgment}
We would like to thank Shuqi (Ariel) Liu for her previous work on establishing a version of Lemma \ref{lem:A_char} for $W_{2,1}^{+}$ and the disjointness portion of Theorem \ref{thm:potpourri}.

\bibliographystyle{IEEEtran}
\bibliography{bib}

\end{document}